\documentclass[preprint,12pt,nopreprintline]{elsarticle}

\usepackage{fix-cm}
\usepackage{lmodern}
\usepackage{amsmath,amssymb,amsthm,bm}
\usepackage{booktabs}
\usepackage{graphicx}
\usepackage{microtype}
\usepackage[hidelinks]{hyperref}

\newtheorem{theorem}{Theorem}
\newtheorem{proposition}[theorem]{Proposition}
\newtheorem{corollary}[theorem]{Corollary}
\theoremstyle{remark}
\newtheorem{remark}[theorem]{Remark}
\newcommand{\dd}{\,\mathrm{d}}
\newcommand{\av}[1]{\langle #1\rangle}
\newcommand{\RR}{\mathbb R}

\allowdisplaybreaks
\begin{document}
\begin{frontmatter}

\title{The monatomic limit of polyatomic gases: the Dirac prescription and an application to a new hierarchy of moments}

\author[unibo,lincei]{Tommaso Ruggeri\texorpdfstring{\fnref{orcid}}{}}
\ead{tommaso.ruggeri@unibo.it}
\fntext[orcid]{ORCID: \href{https://orcid.org/0000-0002-7588-2074}{0000-0002-7588-2074}.}

\affiliation[unibo]{organization={Department of Mathematics, University of Bologna},
            addressline={Via Saragozza 8}, postcode={40123},
            city={Bologna}, country={Italy}}
\affiliation[lincei]{organization={Accademia Nazionale dei Lincei},
            addressline={Via della Lungara 10}, postcode={00165},
            city={Rome}, country={Italy}}

\begin{abstract}
In the kinetic theory of a polytropic gas, internal molecular modes are described by a power-law density of internal states, making the monatomic limit singular. We show that the normalized limit $a\to-1^+$ is equivalent to replacing integration over internal energy by a Dirac mass at zero energy, and derive the first-order correction. We apply this prescription to the classical full-energy moment hierarchy arising from the non-relativistic limit of the relativistic hierarchy of Arima, Carrisi, Pennisi and Ruggeri. For every finite truncation order, we establish the independent and exhaustive monatomic trace identities and derive the number of surviving fields. In particular, the known reduction from polyatomic ET15 to Kremer's monatomic ET14 is recovered. We also prove that the limiting spectrum of the unreduced polyatomic characteristic pencil splits into the spectrum of the reduced monatomic system and a lower-order polyatomic spectrum with a shifted internal-energy parameter. Equilibrium characteristic velocities are obtained using the linearized maximum-entropy prescription, without assuming a nonlinear closure away from equilibrium. We show that, at fixed truncation order, the largest characteristic velocity is non-increasing with the number of molecular degrees of freedom and attains its maximum at the monatomic endpoint. Finally, the symmetric-hyperbolic principal-subsystem argument yields the previously known Boillat--Ruggeri lower bound for the new hierarchy, implying unbounded growth of the maximum characteristic velocity with truncation order.

\end{abstract}


\end{frontmatter}

\section{Introduction}
\label{sec:intro}

For a monatomic rarefied gas the kinetic description is based on the Boltzmann equation
\begin{equation}
 \partial_t f+\xi_i\partial_{x_i}f=Q(f,f),
 \label{eq:Boltzmann-mono}
\end{equation}
where $f=f(t,\bm x,\bm\xi)$ is the molecular distribution function, $\bm\xi$ is the molecular velocity and $Q$ is the collision operator. The transport part of the equation is simple; the entire molecular physics is contained in $Q$, which describes binary collisions and preserves mass, momentum and energy.

For a polyatomic gas the translational velocity alone is not sufficient to specify the molecular state, because part of the molecular energy may be stored in rotational and vibrational modes. Following the continuous-internal-energy description initiated by Borgnakke and Larsen \cite{BorgnakkeLarsen1975} and put into a mathematical kinetic framework by Bourgat, Desvillettes, Le Tallec and Perthame \cite{Bourgat1994}, one introduces a non-negative microscopic internal energy $I$ and considers
\begin{equation}
 f=f(t,\bm x,\bm\xi,I),\qquad I\geq0.
\end{equation}
The kinetic equation keeps the same transport structure,
\begin{equation}
 \partial_t f+\xi_i\partial_{x_i}f=Q[f],
 \label{eq:Boltzmann-poly}
\end{equation}
but the collision term is now an operator in both $\bm\xi$ and $I$: collisions may exchange energy between translation and the rotational and vibrational modes while preserving the total molecular energy. 

The additional integration over $I$ requires a density of internal states $\phi(I)$. More precisely, $\phi(I)\dd I$ is the continuous weight representing the number of internal molecular states whose energy lies between $I$ and $I+\dd I$. Macroscopic quantities are therefore moments of the form
\begin{equation}
 \int_{\RR^3}\int_0^\infty
 \Psi(\bm\xi,I)f(t,\bm x,\bm\xi,I)\,\phi(I)\dd I\dd^3\bm\xi.
 \label{eq:intro-moment}
\end{equation}
For a polytropic gas one takes
\begin{equation}
 \phi(I)=I^a,\qquad a=\frac{D-5}{2},
 \label{eq:phi-a}
\end{equation}
where $D$ is the total number of molecular degrees of freedom. This choice is not merely convenient. At equilibrium the internal partition integral is proportional to $\Gamma(a+1)(k_BT)^{a+1}$, and the mean internal molecular energy is $(a+1)k_BT$. Adding the translational contribution $3k_BT/2$ gives $Dk_BT/2$, which is exactly the caloric law of a polytropic ideal gas.

For non-polytropic gases the specific heat, and hence the caloric internal energy, depends nonlinearly on the temperature. Different kinetic descriptions have been proposed for this situation. Kosuge, Kuo and Aoki introduced a temperature-dependent number of molecular degrees of freedom, leading in the continuous-internal-energy representation to an effective state density depending also on the temperature \cite{KosugeKuoAoki2019}. A related temperature-dependent description was subsequently adopted by Pavi\'c-\v{C}oli\'c and Simi\'c \cite{PavicColicSimic2022}. An alternative approach was developed by Ruggeri and collaborators and, more recently, by Arima, Mentrelli and Ruggeri \cite{ArimaMentrelliRuggeri2024}: the internal state density is kept as a genuine microscopic quantity $\phi(I)$, independent of the macroscopic temperature, and is determined from the prescribed caloric equation of state through an inverse Laplace transform. This permits a nonlinear temperature dependence of the internal energy while preserving the physical interpretation of $\phi(I)\dd I$ as the number of internal states in the interval $[I,I+\dd I]$.

The continuous-internal-energy model gave a natural kinetic basis for extending Rational Extended Thermodynamics (RET) \cite{MullerRuggeri1998} beyond the monatomic gas. Starting with the work \cite{PavicRuggeriSimic2013}, in which the authors applied the Maximum Entropy Principle (MEP), one obtains the classical 14-field theory of a rarefied polyatomic gas. Besides density, velocity and temperature, the fields contain the shear stress, the heat flux and the dynamic pressure $\Pi$, the latter describing the nonequilibrium exchange between translational and internal modes. These developments are collected in the two books by Ruggeri and Sugiyama \cite{RuggeriSugiyama2015,RuggeriSugiyama2021} and the references therein. For a brief history of the MEP in this context, see \cite{ArimaMentrelliRuggeri2024}.

The presence of the dynamic pressure makes the passage to a monatomic gas singular. In a monatomic molecule there is no internal molecular energy and $\Pi$ must disappear. Arima, Taniguchi, Ruggeri and Sugiyama \cite{ArimaEtAl2013} showed that the monatomic equations are recovered only for initial data compatible with the monatomic constraint; in particular,
\begin{equation}
 \Pi(0,\bm x)=0.
 \label{eq:intro-prepared-Pi}
\end{equation}
The limiting equation then preserves $\Pi=0$. The same singular mechanism was subsequently analyzed for the classical many-moment theory \cite{ArimaEtAl2016} and in relativistic models \cite{CarrisiPennisiRuggeri2019,ArimaCarrisi2024}. These results show that the monatomic gas is not obtained by the formal substitution $D=3$ in every constitutive formula: some polyatomic fields cease to be independent and the limiting differential system must be supplied with compatible initial data.

The first classical extension to many moments was organized in two families, the momentum-like $F$ hierarchy and the energy-like $G$ hierarchy. Arima, Mentrelli and Ruggeri \cite{AMR2014} studied the relation between their truncation orders and the corresponding characteristic velocities, and proved that the Boillat--Ruggeri lower bound for the maximum wave velocity remains valid independently of the internal degrees of freedom. In the earlier double hierarchy, the limit of the full polyatomic characteristic spectrum also retains eigenvalues associated with balance equations that become redundant; see Eq.~(23) of \cite{AMRLincei2014}. This observation will provide a useful comparison with the full-energy hierarchy considered here. The singular monatomic reduction of the double hierarchy was later studied in detail in \cite{ArimaEtAl2016}.

A parallel development took place in relativistic kinetic theory. Pennisi and Ruggeri \cite{PennisiRuggeri2017} constructed the first relativistic RET theory of a rarefied polyatomic gas by extending the moment method to a distribution function depending on the additional internal-energy variable. Their construction contains, in the classical limit, the familiar momentum-like $F$ and energy-like $G$ blocks, together with additional hybrid blocks at higher order; the complete classical limit is therefore richer than the $F,G$ pair alone. The subsequent analysis of the non-relativistic limit of relativistic moments by Pennisi and Ruggeri \cite{PennisiRuggeri2020} also clarified an old ambiguity in the classical monatomic theory: once a relativistic hierarchy is truncated at order $N$, its classical limit selects a unique set of moments and has highest polynomial degree
\begin{equation}
 \bar N=2N.
 \label{eq:intro-Nbar}
\end{equation}
Thus the relativistic truncation $N=2$ selects, in the monatomic classical limit, Kremer's 14-field system with the fourth-order scalar moment rather than Grad's 13-field system.

The relativistic construction was then reconsidered from a more direct physical viewpoint. Arima, Carrisi, Pennisi and Ruggeri proposed a new hierarchy in which the $n$th relativistic moment contains the $n$th power of the full molecular energy, the sum of rest energy and internal energy \cite{ACPR2021,ACPREntropy2022}. Its non-relativistic limit was established in \cite{ACPR2022} and yields a precise classical hierarchy with $N+1$ families. The case $N=2$ is a 15-field polyatomic theory, whose reduction to Kremer's ET14 was already discussed in \cite[Sec.~9]{ACPR2021}; the same monatomic endpoint was found in the earlier 15-field construction of \cite{ArimaRuggeri2021}. The present paper considers the full-energy hierarchy at arbitrary finite truncation order.

A separate modelling issue concerns the explicit form of the collision operator $Q$. Besides Boltzmann-type collision operators for polyatomic gases, such as the explicit construction of Djordji\'c, Pavi\'c-\v{C}oli\'c and Torrilhon \cite{DjordjicPavicColicTorrilhon2021}, several simplified kinetic descriptions have been proposed. Representative examples include the ES--BGK model of Brull and Schneider \cite{BrullSchneider2009}, the Fokker--Planck model of Mathiaud and Mieussens \cite{MathiaudMieussens2017}, and the multi-relaxation models of Arima, Ruggeri and Sugiyama in which rotational and vibrational processes are treated separately \cite{ArimaRuggeriSugiyama2017,ArimaRuggeriSugiyama2018}. In the relativistic setting, Pennisi and Ruggeri introduced a BGK model for monatomic and polyatomic gases \cite{PennisiRuggeriBGK2018}; existence, uniqueness and asymptotic behaviour near global equilibrium for the corresponding polyatomic model were subsequently studied by Hwang, Ruggeri and Yun \cite{HwangRuggeriYun2022}. These developments concern the modelling of the collision term and are conceptually distinct from the choice of the internal state density $\phi(I)$.

Our purpose is deliberately focused. For the power-law density \eqref{eq:phi-a}, the monatomic endpoint $D\to3^+$ corresponds to the singular exponent $a\to-1^+$. On the other hand, a genuinely monatomic molecule has no internal mode at all, and the natural kinetic prescription is simply to evaluate every internal-energy integral at $I=0$, formally
\begin{equation}
 \phi(I)\dd I\longrightarrow\delta(I)\dd I.
 \label{eq:formal-delta}
\end{equation}
Here and below the Dirac prescription means the unit mass concentrated at the endpoint $I=0$ of the half-line. This notation refers to the \emph{normalised} internal-energy measure: the unnormalised measure $I^a\dd I$ does not itself converge to a unit Dirac mass as $a\to-1^+$. The concentration at zero internal energy is an elementary limiting mechanism. We give a normalized formulation under explicit assumptions and determine its first-order correction. Our main contribution is an algebraic description valid at every finite truncation order of the full-energy hierarchy: we establish the independent and exhaustive monatomic trace relations and identify the residual spectrum of the singular equilibrium characteristic pencil. The residual block has gamma-shape parameter two and differs from the additional monatomic spectrum found in the earlier $F,G$ hierarchy. The first three truncation orders illustrate these general results. In addition, we prove that the maximum equilibrium characteristic velocity is non-increasing with the number of molecular degrees of freedom at each finite order. Finally, a direct Rayleigh-quotient argument extends the known Boillat--Ruggeri lower bound to the present hierarchy; the inequality itself is not claimed as a new bound.

The paper is organised as follows. Section~\ref{sec:dirac} proves the equivalence between $a\to-1^+$ and the Dirac prescription. Section~\ref{sec:origin} recalls explicitly the relativistic hierarchy from which the classical full-energy hierarchy is obtained. Section~\ref{sec:reduction} gives the monatomic reduction and the redundant balance equations. Section~\ref{sec:examples} discusses the first three truncation orders, and Section~\ref{sec:speeds} examines the characteristic velocities, their monotonic dependence on the molecular degrees of freedom, and their growth with the truncation order.

\section{\texorpdfstring{The limit $a\to-1^+$ and the Dirac prescription}{The limit a to -1+ and the Dirac prescription}}
\label{sec:dirac}

Put
\begin{equation}
        \alpha=a+1=\frac{D-3}{2}>0.
 \label{eq:alpha}
\end{equation}
Thus the monatomic limit is $\alpha\to0^+$. Let $w$ be a non-negative continuous function on $[0,\infty)$, with $w(0)>0$. Assume that, for some fixed $\alpha_0>0$,
\begin{equation}
 \int_1^\infty I^{\alpha_0-1}w(I)\dd I<\infty.
 \label{eq:tailw}
\end{equation}
For the rest of this section we restrict to $0<\alpha\leq\min(1,\alpha_0)$, which is the only range needed for the limit. Define
\begin{equation}
 Z_\alpha=\int_0^\infty I^{\alpha-1}w(I)\dd I.
 \label{eq:Zalpha}
\end{equation}

\begin{theorem}[Normalized Dirac limit]
\label{thm:dirac}
Under the assumptions above,
\begin{equation}
       \alpha Z_\alpha\longrightarrow w(0).
 \label{eq:Zlimit}
\end{equation}
Let $G$ be continuous at $I=0$ and suppose that
\begin{equation}
 \int_0^1 \frac{|G(I)-G(0)|}{I}w(I)\dd I
 +\int_1^\infty |G(I)-G(0)|I^{\alpha_0-1}w(I)\dd I<\infty.
 \label{eq:Gassumption}
\end{equation}
Then, with $\delta_0$ denoting the unit Dirac mass at the endpoint $I=0$ of $[0,\infty)$,
\begin{equation}
 \boxed{
 \lim_{a\to-1^+}
 \frac{\displaystyle\int_0^\infty G(I)I^a w(I)\dd I}
      {\displaystyle\int_0^\infty I^a w(I)\dd I}
 =G(0)
 =\int_{[0,\infty)} G(I)\,\delta_0(\dd I).}
 \label{eq:main-equivalence}
\end{equation}
Moreover,
\begin{equation}
 \frac{\displaystyle\int_0^\infty G(I)I^a w(I)\dd I}
      {\displaystyle\int_0^\infty I^a w(I)\dd I}
 =G(0)+\frac{a+1}{w(0)}
 \int_0^\infty\frac{G(I)-G(0)}{I}w(I)\dd I
 +o(a+1).
 \label{eq:firstcorrection}
\end{equation}
\end{theorem}

\begin{proof}
We divide the proof into three elementary steps.

\emph{Step 1: the singular part of the normalization comes only from $I=0$.}
Given $\varepsilon>0$, continuity of $w$ at the origin allows us to choose and then fix $b\in(0,1)$ such that
\[
 w(0)-\varepsilon\leq w(I)\leq w(0)+\varepsilon,
 \qquad 0<I<b.
\]
Therefore
\[
 (w(0)-\varepsilon)\frac{b^\alpha}{\alpha}
 \leq
 \int_0^b I^{\alpha-1}w(I)\dd I
 \leq
 (w(0)+\varepsilon)\frac{b^\alpha}{\alpha}.
\]
After multiplication by $\alpha$ this becomes
\[
 (w(0)-\varepsilon)b^\alpha
 \leq
 \alpha\int_0^b I^{\alpha-1}w(I)\dd I
 \leq
 (w(0)+\varepsilon)b^\alpha.
\]
For fixed $b$, $b^\alpha\to1$ as $\alpha\to0^+$. Hence the contribution from $(0,b)$ approaches $w(0)$, up to an arbitrarily small error $\varepsilon$.

\emph{Step 2: the part away from the origin disappears after multiplication by $\alpha$.}
On the compact interval $[b,1]$, $w$ is bounded, say by $C_b$. Since $0<\alpha\leq1$,
\[
 0\leq \alpha\int_b^1 I^{\alpha-1}w(I)\dd I
 \leq \alpha C_b\int_b^1\frac{\dd I}{I}
 =\alpha C_b\log\frac1b\longrightarrow0.
\]
On $[1,\infty)$, for $0<\alpha\leq\alpha_0$,
\[
 I^{\alpha-1}\leq I^{\alpha_0-1},
\]
and assumption \eqref{eq:tailw} gives
\[
 0\leq \alpha\int_1^\infty I^{\alpha-1}w(I)\dd I
 \leq \alpha\int_1^\infty I^{\alpha_0-1}w(I)\dd I
 \longrightarrow0.
\]
Combining the three intervals $(0,b)$, $[b,1]$ and $[1,\infty)$, and then letting $\varepsilon\to0$, yields
\[
 \alpha Z_\alpha\longrightarrow w(0),
\]
which proves \eqref{eq:Zlimit}. Equivalently,
\begin{equation}
 Z_\alpha=\frac{w(0)}{\alpha}+o(\alpha^{-1}),
 \qquad
 \frac1{Z_\alpha}=\frac{\alpha}{w(0)}+o(\alpha).
 \label{eq:Zinverse}
\end{equation}

\emph{Step 3: normalized averages concentrate at $I=0$.}
Write
\[
 \int_0^\infty G(I)I^{\alpha-1}w(I)\dd I
 =G(0)Z_\alpha+M_\alpha,
\]
where
\[
 M_\alpha=
 \int_0^\infty [G(I)-G(0)]I^{\alpha-1}w(I)\dd I.
\]
For $0<I<1$ and $0<\alpha\leq1$,
\[
 |G(I)-G(0)|I^{\alpha-1}w(I)
 \leq \frac{|G(I)-G(0)|}{I}w(I),
\]
while for $I\geq1$ and $0<\alpha\leq\alpha_0$,
\[
 |G(I)-G(0)|I^{\alpha-1}w(I)
 \leq |G(I)-G(0)|I^{\alpha_0-1}w(I).
\]
Both majorants are integrable by \eqref{eq:Gassumption}. Thus dominated convergence gives
\[
 M_\alpha\longrightarrow
 M_0=\int_0^\infty\frac{G(I)-G(0)}{I}w(I)\dd I.
\]
Consequently,
\[
 \frac{\int_0^\infty G(I)I^{\alpha-1}w(I)\dd I}{Z_\alpha}
 =G(0)+\frac{M_\alpha}{Z_\alpha}.
\]
Using \eqref{eq:Zinverse} and $M_\alpha=M_0+o(1)$,
\[
 \frac{M_\alpha}{Z_\alpha}
 =\frac{\alpha}{w(0)}M_0+o(\alpha).
\]
Since $\alpha=a+1$, this proves both the Dirac limit \eqref{eq:main-equivalence} and the first-order expansion \eqref{eq:firstcorrection}.
\end{proof}

Equation \eqref{eq:main-equivalence} gives the precise equivalence, for normalized kinetic moments, between the singular limit $a\to-1^+$ and the prescription \eqref{eq:formal-delta}.

\subsection*{The classical equilibrium as an immediate example}
Let $\theta=k_BT$, where $k_B$ is Boltzmann's constant and $T$ is the absolute temperature, and take $w(I)=e^{-I/\theta}$. Then
\begin{equation}
 Z_\alpha=\Gamma(\alpha)\theta^\alpha,
 \label{eq:gammaZ}
\end{equation}
and the normalised internal-energy integrals are
\begin{equation}
 \av{G}_\alpha=
 \frac{\displaystyle\int_0^\infty G(I)I^{\alpha-1}e^{-I/\theta}\dd I}
      {\Gamma(\alpha)\theta^\alpha}.
 \label{eq:gammamean}
\end{equation}
Therefore \eqref{eq:main-equivalence} gives directly $\av G_\alpha\to G(0)$. In particular, for every integer $r\geq1$,
\begin{equation}
 \av{I^r}_\alpha
 =\theta^r\frac{\Gamma(\alpha+r)}{\Gamma(\alpha)}
 =\alpha(r-1)!\theta^r+O(\alpha^2).
 \label{eq:gamma-moments}
\end{equation}
It is useful to connect this elementary formula with the recurrence relation introduced for the general non-polytropic hierarchy in \cite{ACPR2021}. If
\begin{equation}
 \bar A_r=\left\langle\left(\frac{I}{k_BT}\right)^r\right\rangle_\alpha
 =\frac{\Gamma(\alpha+r)}{\Gamma(\alpha)},
\end{equation}
then the Gamma recurrence immediately gives
\begin{equation}
 \bar A_{r+1}=(\alpha+r)\bar A_r,\qquad \bar A_0=1.
 \label{eq:polytropic-recurrence}
\end{equation}
In the notation of \cite[Sec.~4]{ACPR2021}, the corresponding recurrence for a general internal-state density can be written as
\begin{equation}
 \bar A_{r+1}=T\frac{\dd \bar A_r}{\dd T}+
 (r+\bar\varepsilon_I)\bar A_r,
 \qquad
 \bar\varepsilon_I=\frac{m\varepsilon_I}{k_BT},
 \label{eq:general-recurrence}
\end{equation}
where $\varepsilon_I$ is the equilibrium specific energy carried by the internal molecular modes. For a polytropic gas, $m\varepsilon_I=\alpha k_BT$, so that $\bar\varepsilon_I=\alpha$ and the coefficients $\bar A_r$ are independent of $T$; hence \eqref{eq:general-recurrence} reduces exactly to \eqref{eq:polytropic-recurrence}. Thus the Gamma identity used here is the polytropic specialization of the recurrence already obtained for the more general internal-energy integrals.
The mean molecular energy is consequently
\begin{equation}
 m\varepsilon=\left(\frac32+\alpha\right)k_BT=\frac D2 k_BT,
 \label{eq:classical-energy}
\end{equation}
and its monatomic value is obtained without any separate limiting calculation.

\section{Relativistic origin of the new classical hierarchy}
\label{sec:origin}

The classical hierarchy used below is not chosen ad hoc. It is the precise non-relativistic limit of a relativistic hierarchy, and for this reason it is useful to display the relativistic system explicitly.

The first relativistic RET theory of a rarefied polyatomic gas was constructed by Pennisi and Ruggeri \cite{PennisiRuggeri2017}. In that formulation the generalized moments were chosen so that the non-relativistic limit reproduces the earlier classical $F,G$ hierarchy. The more recent full-energy hierarchy of Arima, Carrisi, Pennisi and Ruggeri \cite{ACPREntropy2022,ACPR2022} replaces that construction by moments containing powers of the actual total molecular energy $mc^2+I$.

Let $p^\mu$ be the molecular four-momentum, let $m$ be the molecular mass and let $c$ denote the speed of light. The relativistic balance laws are
\begin{equation}
 \partial_\mu A^{\mu\mu_1\cdots\mu_n}
 =\mathcal J^{\mu_1\cdots\mu_n},
 \qquad n=0,1,\ldots,N,
 \label{eq:rel-balances}
\end{equation}
where
\begin{align}
 A^{\mu\mu_1\cdots\mu_n}
 &=\left(\frac1{mc}\right)^{2n-1}
 \int_{\RR^3}\int_0^\infty
 f\,p^\mu p^{\mu_1}\cdots p^{\mu_n}
 (mc^2+I)^n\phi(I)\dd I\dd P,
 \label{eq:rel-moments}\\
 \mathcal J^{\mu_1\cdots\mu_n}
 &=\left(\frac1{mc}\right)^{2n-1}
 \int_{\RR^3}\int_0^\infty
 Q\,p^{\mu_1}\cdots p^{\mu_n}
 (mc^2+I)^n\phi(I)\dd I\dd P.
 \label{eq:rel-productions}
\end{align}
Here $\dd P=\dd^3\bm p/p^0$ is the Lorentz-invariant momentum measure. In the standard inertial parametrization used in \cite{ACPR2022}, $x^0=ct$, $p^0=mc\gamma_{\mathrm L}$, $p^i=m\gamma_{\mathrm L}\xi_i$, with $\gamma_{\mathrm L}=(1-\xi^2/c^2)^{-1/2}$. The normalization of the distribution functions is the one adopted there, for which the relativistic distribution tends to the classical one with the corresponding factor $m^{-3}$ in the non-relativistic limit. The symbol $\mathcal J^{\mu_1\cdots\mu_n}$ denotes a relativistic production moment, while $I$ denotes the molecular internal energy. Greek spacetime indices take the values $0,1,2,3$. The index $\mu$ in \eqref{eq:rel-balances} is the divergence, or flux, index. The truncation index $N$ counts only the remaining tensorial indices $\mu_1,\ldots,\mu_n$: thus the hierarchy is truncated by imposing $n\leq N$. This convention is important below when $N$ is compared with the classical tensorial orders.

The non-relativistic limit $c\to\infty$, performed through the finite combinations described in \cite{ACPR2022}, gives the $N+1$ classical families
\begin{equation}
 \begin{aligned}
 H_A^{(s)}&=m\int_{\RR^3}\int_0^\infty
 f\,\xi_A\left(\xi^2+\frac{2I}{m}\right)^s
 \phi(I)\dd I\dd^3\bm\xi,\\
 &\hspace{16mm}s=0,\ldots,N,\qquad A=0,\ldots,N-s.
 \end{aligned}
 \label{eq:H-hierarchy}
\end{equation}
The corresponding balance equations are
\begin{equation}
 \partial_t H_A^{(s)}+\partial_{x_j}H_{jA}^{(s)}=J_A^{(s)},
 \qquad s+A\leq N,
 \label{eq:H-balances}
\end{equation}
where the collision moments are explicitly given by
\begin{equation}
 J_A^{(s)}=m\int_{\RR^3}\int_0^\infty
 Q[f] \,\xi_A\left(\xi^2+\frac{2I}{m}\right)^s
 \phi(I)\dd I\dd^3\bm\xi.
 \label{eq:collision-moments}
\end{equation}
For a collision operator conserving mass, momentum and total molecular energy,
\begin{equation}
 J^{(0)}=0,\qquad J_i^{(0)}=0,\qquad J^{(1)}=0.
 \label{eq:collision-invariants}
\end{equation}
All tensor moments are symmetric in their free spatial indices, and only their independent components are counted. In \eqref{eq:H-hierarchy}, $A$ is the number of free spatial tensor indices. For $A\geq1$,
\[
 \xi_A=\xi_{i_1}\cdots\xi_{i_A},
\]
and, importantly, for $A=0$ we set
\begin{equation}
 \xi_A=1.
 \label{eq:A0-convention}
\end{equation}
Thus no molecular-velocity factor is present in a scalar moment.

The highest polynomial degree in the molecular velocity among all the classical moments is $\bar N=2N$, in agreement with \eqref{eq:intro-Nbar}. This $\bar N$ must not be confused with the highest tensorial order $N$ of the purely translational $s=0$ block.

The first two blocks are the familiar momentum- and energy-like moments,
\begin{equation}
 F_A=H_A^{(0)},\qquad G_{llA}=H_A^{(1)},
 \label{eq:FG}
\end{equation}
while the levels $s\geq2$ are the genuinely additional moments generated by the relativistic full-energy hierarchy. The number of independent components at fixed $N$ is
\begin{equation}
 d_N=\sum_{s=0}^N\sum_{A=0}^{N-s}\binom{A+2}{2}
 =\binom{N+4}{4},
 \label{eq:poly-count}
\end{equation}
so that $d_1=5$, $d_2=15$ and $d_3=35$.

\section{Monatomic reduction of the hierarchy}
\label{sec:reduction}

The monatomic reduction follows directly from the disappearance of the internal molecular energy. For every $s$ and $A$ such that both moments belong to the truncated hierarchy,
\begin{align}
 H_A^{(s+1)}-H_{llA}^{(s)}
 &=m\int\!\!\int f\,\xi_A
 \left(\xi^2+\frac{2I}{m}\right)^s
 \left(\frac{2I}{m}\right)\phi(I)\dd I\dd^3\bm\xi\nonumber\\
 &=2\int\!\!\int I f\,\xi_A
 \left(\xi^2+\frac{2I}{m}\right)^s
 \phi(I)\dd I\dd^3\bm\xi.
 \label{eq:defect}
\end{align}
The difference between two successive energy levels is therefore a moment containing an explicit factor $I$. Under the direct Dirac prescription it vanishes identically and we obtain
\begin{equation}
 H_A^{(s+1)}=H_{llA}^{(s)},\qquad s+A\leq N-2.
 \label{eq:trace-relations}
\end{equation}
For equilibrium moments the same relations follow from the polytropic limit $a\to-1^+$ by Theorem~\ref{thm:dirac}.

It is useful to write also the balance equations associated with these relations. Subtracting the trace of the $s$th equation from the $(s+1)$st equation in \eqref{eq:H-balances} gives exactly
\begin{align}
 &\partial_t\left(H_A^{(s+1)}-H_{llA}^{(s)}\right)
 +\partial_{x_j}\left(H_{jA}^{(s+1)}-H_{jllA}^{(s)}\right)\nonumber\\
 &\hspace{35mm}=J_A^{(s+1)}-J_{llA}^{(s)},
 \qquad s+A\leq N-2.
 \label{eq:defect-balance}
\end{align}
Every density, flux and production in \eqref{eq:defect-balance} contains the same explicit internal-energy factor as in \eqref{eq:defect}. Consequently, for the direct monatomic prescription $I=0$ the entire equation becomes the identity $0=0$. The remaining equations are precisely the independent monatomic balance laws.

The number of independent relations in \eqref{eq:trace-relations} is
\begin{equation}
 k_N=\sum_{s=0}^{N-2}\sum_{A=0}^{N-s-2}\binom{A+2}{2}
 =\binom{N+2}{4}.
 \label{eq:relation-count}
\end{equation}
It remains to check that these relations are independent and exhaustive. This can be seen by a short polynomial argument. Introduce a formal variable $e$ standing for the factor $\xi^2+2I/m$, and let $V_N$ be the space spanned by the monomials $\xi_A e^s$ with $s+A\leq N$. At the monatomic endpoint one sets $I=0$, hence $e=\xi^2$. Consider therefore the restriction map
\[
 R:V_N\longrightarrow\mathbb R[\xi_1,\xi_2,\xi_3],
 \qquad R(P)(\bm\xi)=P(\bm\xi,\xi^2).
\]
Since $e-\xi^2$ is monic in $e$, polynomial division with remainder in the variable $e$ shows that every $P\in\ker R$ can be written as $P=(e-\xi^2)S$ (the remainder is zero because $R(P)=0$). Here the degree defining $V_N$ is the ordinary total degree in the four formal variables $(\xi_1,\xi_2,\xi_3,e)$, each assigned degree one. If $S\ne0$ has highest homogeneous part $S_d$, the product has the nonzero highest homogeneous part $-\xi^2S_d$ of degree $d+2$. Consequently,
\[
 \deg\bigl((e-\xi^2)S\bigr)=\deg S+2,
\]
so that $S\in V_{N-2}$ whenever $(e-\xi^2)S\in V_N$. Conversely, every element of $(e-\xi^2)V_{N-2}$ belongs to $V_N$ and vanishes under $R$. Thus
\begin{equation}
 \ker R=(e-\xi^2)V_{N-2}.
 \label{eq:reduction-kernel}
\end{equation}
Multiplication by $e-\xi^2$ is injective on $V_{N-2}$, so the kernel has dimension $\dim V_{N-2}=\binom{N+2}{4}$. Its natural generators are precisely
\[
 (e-\xi^2)\xi_Ae^s,\qquad s+A\leq N-2,
\]
which are the polynomial counterparts of \eqref{eq:trace-relations}. Thus the trace relations listed above are all the relations and are linearly independent.

Hence the number of independent monatomic fields is
\begin{equation}
 r_N=d_N-k_N
 =\binom{N+4}{4}-\binom{N+2}{4}
 =\frac{(N+1)(N+2)(2N+3)}6.
 \label{eq:mono-count}
\end{equation}
For the first orders,
\begin{center}
\begin{tabular}{cccc}
\toprule
$N$ & polyatomic fields $d_N$ & relations $k_N$ & monatomic fields $r_N$\\
\midrule
1&5&0&5\\
2&15&1&14\\
3&35&5&30\\
4&70&15&55\\
\bottomrule
\end{tabular}
\end{center}

\begin{remark}[Prepared initial data and redundant equations]
As is typical of a singular limit in which the limiting system has fewer
independent variables, the initial data must be \emph{prepared}, that is,
compatible with the constraints of the limiting system. In the present
problem this means that, for a family of polyatomic gases with $D>3$
approaching the monatomic value $D=3$, the combinations that become
identically zero in the monatomic theory must tend to zero in the initial data.

The trace identities
\begin{equation}
 H_A^{(s+1)}=H_{llA}^{(s)},\qquad s+A\leq N-2,
 \label{eq:remark-trace}
\end{equation}
characterize the monatomic system itself and must not be imposed as exact
identities for a polyatomic gas with fixed $D>3$. Indeed, already at
equilibrium,
\begin{equation}
 H_E^{(1)}-H_{ll,E}^{(0)}=(D-3)p=2\alpha p>0.
 \label{eq:remark-equilibrium-defect}
\end{equation}
More generally, \eqref{eq:defect} shows that this defect measures the
contribution of the molecular internal energy.

Accordingly, the natural prepared-data condition for a family of
polyatomic initial states is
\begin{equation}
 \begin{aligned}
 &H_A^{(s+1)}(0,\bm x;D)-H_{llA}^{(s)}(0,\bm x;D)
 \longrightarrow0,\\
 &\hspace{15mm}D\to3^+,\qquad s+A\leq N-2,
 \end{aligned}
 \label{eq:prepared-initial}
\end{equation}
in the topology appropriate to the convergence problem. For the limiting
monatomic initial data, the corresponding trace identities hold exactly.
In the classical ET14 case this compatibility condition reduces to
\begin{equation}
 \Pi(0,\bm x)=0,
 \label{eq:remark-Pi}
\end{equation}
as shown in \cite{ArimaEtAl2013}; analogous prepared-data conditions occur
for the earlier $F,G$ hierarchy \cite{ArimaEtAl2016}.

This compatibility is not by itself a sufficient condition for the
convergence of polyatomic solutions. To illustrate the distinction
already at the linearized kinetic level, let $y$ have the normalized Gamma distribution of shape $\alpha$ and unit scale, with density $e^{-y}y^{\alpha-1}/\Gamma(\alpha)$ for $y>0$, and set
\[
 h_\alpha=\frac{y-\alpha}{\sqrt{\alpha}}.
\]
Then
\[
 \langle h_\alpha^2\rangle_\alpha=1,
 \qquad
 \langle yh_\alpha\rangle_\alpha=\sqrt{\alpha}\longrightarrow0.
\]
Thus a vanishing unscaled internal-energy defect need not control the
corresponding rescaled entropy amplitude. This example is not a claim
about solutions of the nonlinear balance laws.

Equation \eqref{eq:defect-balance} identifies the balance laws that become
redundant in the monatomic limit. Under the direct Dirac prescription,
their densities, fluxes and productions vanish identically and the
corresponding equations reduce to $0=0$. This algebraic reduction should
be distinguished from the propagation of the limiting constraints, which
requires the closed limiting system and uniqueness of its solution, and
from the still stronger problem of proving convergence of solutions of
the polyatomic systems as $D\to3^+$.
\end{remark}

\section{The first three truncation orders}
\label{sec:examples}

\subsection{\texorpdfstring{$N=1$: the five-field Euler system}{N=1: the five-field Euler system}}
For $N=1$ the moments are
\[
 H^{(0)},\qquad H_i^{(0)},\qquad H^{(1)}.
\]
They represent mass density, momentum density and total energy density. Here $\rho$ is the mass density, $v_i$ are the components of the macroscopic gas velocity, and $\varepsilon$ is the specific internal energy. In physical variables,
\begin{equation}
 H^{(0)}=\rho,\qquad
 H_i^{(0)}=\rho v_i,\qquad
 H^{(1)}=\rho v^2+2\rho\varepsilon.
 \label{eq:N1-physical}
\end{equation}
There is no relation of the form \eqref{eq:trace-relations}, because $N-2<0$. The polyatomic and monatomic systems therefore have the same five fields; only the caloric equation of state changes from $D$ degrees of freedom to $D=3$.

\subsection{\texorpdfstring{$N=2$: the 15-field system and the Kremer limit}{N=2: the 15-field system and the Kremer limit}}
For $N=2$ the hierarchy contains
\begin{equation}
 H^{(0)},\ H_i^{(0)},\ H_{ij}^{(0)},\qquad
 H^{(1)},\ H_i^{(1)},\qquad H^{(2)},
 \label{eq:N2-moments}
\end{equation}
that is $1+3+6+1+3+1=15$ fields. In the local rest frame, where $v_i=0$, we write
\begin{align}
 H^{(0)}&=\rho,\qquad H_i^{(0)}=0,\nonumber\\
 H_{ij}^{(0)}&=(p+\Pi)\delta_{ij}-\sigma_{\langle ij\rangle},\label{eq:N2-stress}\\
 H^{(1)}&=2\rho\varepsilon,\qquad H_i^{(1)}=2q_i.\label{eq:N2-energy}
\end{align}
Here $p$ is the equilibrium pressure, $\Pi$ the dynamic pressure, $\delta_{ij}$ the Kronecker symbol, $\sigma_{\langle ij\rangle}$ the symmetric traceless shear-stress tensor, and $q_i$ the heat-flux vector. The brackets $\langle ij\rangle$ denote the symmetric traceless part.
The additional scalar is written
\begin{equation}
 H^{(2)}=H_E^{(2)}+\Delta,
 \qquad
 H_E^{(2)}=\rho D(D+2)\left(\frac{k_BT}{m}\right)^2.
 \label{eq:Delta-def}
\end{equation}
where the subscript $E$ denotes the equilibrium value and $\Delta$ is its nonequilibrium part. With this normalization, when $D\to3^+$ one has
\begin{equation}
 H^{(2)}\longrightarrow F_{llkk}
 =15\frac{p^2}{\rho}+\Delta,
 \label{eq:Kremer-Delta}
\end{equation}
so that the limiting $\Delta$ is exactly the nonequilibrium scalar introduced by Kremer in the monatomic 14-field theory \cite{Kremer1986,RuggeriSugiyama2021}.

The 15 physical variables may therefore be represented by
\begin{equation}
 \rho,\quad v_i,\quad T,\quad \Pi,\quad
 \sigma_{\langle ij\rangle},\quad q_i,\quad \Delta.
 \label{eq:15fields}
\end{equation}
The only monatomic relation at this order is
\begin{equation}
              H^{(1)}=H_{ll}^{(0)}.
 \label{eq:N2-relation}
\end{equation}
In the rest frame it gives $2\rho\varepsilon=3(p+\Pi)$. Since a monatomic ideal gas satisfies $2\rho\varepsilon=3p$, the compatible monatomic solution has
\begin{equation}
                 \Pi=0.
 \label{eq:Pi-zero}
\end{equation}
Consequently, at the level of the independent moments and of their linearized MEP closure,
\begin{equation}
 \text{polyatomic ET15}\quad\xrightarrow[D\to3^+]{}\quad
        \text{Kremer monatomic ET14},
 \label{eq:ET15-Kremer}
\end{equation}
with fields $\rho,v_i,T,\sigma_{\langle ij\rangle},q_i,\Delta$. For the full-energy ET15 closure this monatomic reduction is already discussed in \cite[Sec.~9, Eqs.~(61)--(62)]{ACPR2021}. The reduction of the production terms additionally depends on the adopted collision model; it is not a consequence of the field count $15\to14$ alone. The same monatomic endpoint was found in the earlier 15-field construction of \cite{ArimaRuggeri2021}; it is not the Grad 13-field system.

\subsection{\texorpdfstring{$N=3$: the $35\to30$ reduction}{N=3: the 35 to 30 reduction}}
For $N=3$ the hierarchy contains 35 fields. Equation \eqref{eq:trace-relations} gives five independent relations,
\begin{equation}
 H^{(1)}=H_{ll}^{(0)},\qquad
 H_i^{(1)}=H_{lli}^{(0)}\quad(i=1,2,3),\qquad
 H^{(2)}=H_{ll}^{(1)}.
 \label{eq:N3-relations}
\end{equation}
The first relation again removes the dynamic pressure. The next three identify the energy-flux vector with the trace of the third-order momentum tensor, and the last one identifies the scalar fourth-order moment with the trace of the next energy level. Hence 35 polyatomic fields reduce to 30 independent monatomic fields. At this order it is more transparent to retain the moment notation than to introduce special names for all the higher-order fields.

\section{Equilibrium characteristic velocities}
\label{sec:speeds}

We now study the equilibrium characteristic velocities of the new hierarchy. Our aims are to describe the first truncation orders, prove the monotonic dependence of the maximum characteristic velocity on the molecular degrees of freedom at fixed order, and investigate its growth with the truncation order. Throughout this section the equilibrium state is taken in the local rest frame. Thus $v_i$ continues to denote only the macroscopic velocity of the gas, and we denote by $\lambda$ a characteristic velocity.

\subsection*{Equilibrium distribution and the linearized MEP closure}
The hierarchy \eqref{eq:H-balances} is an unclosed moment system: in general,
its fluxes $H_{jA}^{(s)}$ are not prescribed functions of the retained
moments. We use the Maximum Entropy Principle (MEP) only at the level of
its \emph{equilibrium linearization}, as in the moment calculations of
\cite{AMR2014,ACPR2021,ArimaRuggeri2021}. This distinction is essential:
replacing $f$ by its equilibrium value in the nonlinear flux integrals
would instead give only equilibrium fluxes, not the characteristic
velocities of the retained nonequilibrium fields.

For the polytropic state density $\phi(I)=I^{\alpha-1}$, the equilibrium
distribution, normalized so that
$m\int f_E\phi(I)\dd I\dd^3\bm\xi=\rho$, is
\begin{equation}
 \begin{aligned}
 f_E(\bm\xi,I)&=\frac{\rho}{m}
 \left(\frac{m}{2\pi k_BT}\right)^{3/2}
 \frac{1}{\Gamma(\alpha)(k_BT)^\alpha}\\
 &\qquad\times
 \exp\left[-\frac{m|\bm\xi-\bm v|^2}{2k_BT}
                 -\frac{I}{k_BT}\right],\qquad \alpha>0.
 \end{aligned}
 \label{eq:explicit-fE}
\end{equation}
Here $\bm v$ is the macroscopic velocity and $\bm\xi$ the molecular velocity.
The measure in the entropy functional and in all moments remains
$\phi(I)\dd I\dd^3\bm\xi$; in particular, $\phi$ is \emph{not} absorbed
into $f_E$.

Only the equilibrium characteristic velocities are considered here. We
therefore need only the linearized Maximum Entropy Principle, not a
nonlinear expression for the fluxes in terms of the retained moments.
Let $\Psi_\beta$ denote an independent list of the retained
 moment polynomials. Writing $\delta f=f_Eh$, its variational characterization is the
minimization of the quadratic entropy variation
\begin{equation}
 \frac12\int_{\RR^3}\int_0^\infty
      f_Eh^2\phi(I)\dd I\dd^3\bm\xi
 \label{eq:quadratic-MEP}
\end{equation}
under prescribed linearized moment constraints. By orthogonal
projection in $L^2(f_E\phi\dd I\dd^3\bm\xi)$, the unique minimizer is
\begin{equation}
 h=\sum_\beta z_\beta\Psi_\beta,
 \label{eq:MEP-tangent}
\end{equation}
where the coefficients $z_\beta$ are determined by the prescribed linearized moments. Hence the density and flux variations involve equilibrium
averages of products of these polynomials, with an additional factor
$\xi_1$ for the $x_1$-flux. All such Gaussian--gamma integrals are
finite for every finite truncation order $N$, including odd $N$.
This tangent construction is well defined without assuming the
existence of a nonlinear exponential MEP closure on a whole
neighborhood of equilibrium. In particular, the equilibrium
Maxwellian is used to evaluate the first variations of \emph{both}
densities and fluxes, not to replace the fluxes by their equilibrium
values.

Let
\begin{equation}
 v_T=\sqrt{\frac{2k_BT}{m}},\qquad
 \eta_i=\frac{\xi_i}{v_T},\qquad
 y=\frac{I}{k_BT},\qquad
 \mathcal E=\eta^2+y,
 \label{eq:scaled-vars}
\end{equation}
where $m$ is the molecular mass and $k_B$ is Boltzmann's constant. For the chosen equilibrium state $T$ is fixed, and therefore $v_T$ is a constant reference thermal speed in the linearized characteristic calculation. It is not a new field and must not be confused with the macroscopic gas velocity $v_i$. Moreover,
\begin{equation}
 \frac{m\xi^2}{2}+I=k_BT\,\mathcal E,
 \label{eq:dimensionless-energy}
\end{equation}
so that $\mathcal E$ is the total molecular energy measured in units of $k_BT$.

Recall that
\begin{equation}
 \alpha=a+1=\frac{D-3}{2}>0.
\end{equation}
At equilibrium we introduce the normalized average
\begin{equation}
 \langle g\rangle_\alpha=
 \frac{1}{\pi^{3/2}\Gamma(\alpha)}
 \int_{\RR^3}\int_0^\infty
 g(\bm\eta,y)e^{-\eta^2-y}y^{\alpha-1}\dd y\dd^3\bm\eta,
 \qquad \langle1\rangle_\alpha=1.
 \label{eq:eq-average}
\end{equation}
For the moment $H_A^{(s)}$ the corresponding dimensionless polynomial is
\begin{equation}
 P_A^{(s)}(\bm\eta,y)=\eta_A\mathcal E^s,
 \qquad s+A\leq N,
 \label{eq:moment-polynomials}
\end{equation}
with $\eta_A=\eta_{i_1}\cdots\eta_{i_A}$ and $\eta_0=1$, exactly as for $\xi_A$ in Section~\ref{sec:origin}.

For clarity, rescale each physical moment by the fixed equilibrium
factor $v_T^{A+2s}$. In the rest frame and with
$h=\sum_\beta z_\beta\Psi_\beta$ expressed in the dimensionless polynomial basis $\{\Psi_\beta\}$,
its first variations take the form
\begin{equation}
 \delta U_P=\rho\langle Ph\rangle_\alpha,\qquad
 \delta F_P^1=\rho v_T\langle\eta_1Ph\rangle_\alpha.
 \label{eq:linearized-density-flux}
\end{equation}
Consequently, both matrices below are the coefficients of the
linearized moment and flux maps generated by $\delta f=f_Eh$.

Consider a plane wave propagating in the $x_1$-direction. Let
$\mathcal P_N$ denote the finite set of all polynomials
\eqref{eq:moment-polynomials}, including every independent symmetric
tensor component. After arranging the elements of $\mathcal P_N$ in
any fixed order, the time matrix $\mathcal A$ and the flux matrix
$\mathcal B$ are defined by the two bilinear forms
\begin{align}
 \mathcal A(\Psi_\beta,\Psi_\gamma)&=\langle\Psi_\beta\Psi_\gamma\rangle_\alpha,
 \label{eq:time-matrix}\\
 \mathcal B(\Psi_\beta,\Psi_\gamma)&=\langle\eta_1\Psi_\beta\Psi_\gamma\rangle_\alpha,
 \label{eq:flux-matrix}
\end{align}
for $\Psi_\beta,\Psi_\gamma\in\mathcal P_N$.

For completeness, this characteristic construction has a direct
quadratic variational proof. Enumerate the retained independent
polynomials as $P_1,\ldots,P_{d_N}$, write
$b_\beta=\delta U_{P_\beta}/\rho$, and minimize
\begin{equation}
 \frac12\langle h^2\rangle_\alpha
 \quad\text{subject to}\quad
 \langle P_\beta h\rangle_\alpha=b_\beta
 \quad(\beta=1,\ldots,d_N).
 \label{eq:MEP-quadratic-minimum}
\end{equation}
For $\alpha>0$ the polynomials $P_\beta$ are linearly independent in
the Gaussian--gamma $L^2$ space, whose measure has positive density on
an open set. Thus $\mathcal A$ is positive definite. Defining
$z=\mathcal A^{-1}b$, the unique candidate
$h_*=\sum_\beta z_\beta P_\beta$ satisfies the constraints.
Every other admissible $h$ has the form $h=h_*+g$ with
$\langle P_\beta g\rangle_\alpha=0$ for each $\beta$, and hence
\begin{equation}
 \langle h^2\rangle_\alpha
 =\langle h_*^2\rangle_\alpha+\langle g^2\rangle_\alpha.
 \label{eq:MEP-projection-identity}
\end{equation}
This proves uniqueness and gives
$\delta U=\rho\mathcal A z$ and
$\delta F^1=\rho v_T\mathcal Bz$. As every polynomial moment is
integrable against the Gaussian--gamma equilibrium weight, the proof
holds for each finite $N$, independently of its parity, without
assuming a nonlinear exponential MEP closure in a full equilibrium
neighborhood.

For $D>3$, after retaining only independent symmetric tensor components, the matrix $\mathcal A$ is positive definite and both matrices are symmetric. The coefficients $z$ of the MEP tangent perturbation are the linearized entropy variables (main fields, up to the chosen normalization and sign convention). The resulting equilibrium tangent system is symmetric hyperbolic, in accordance with the main-field framework of Ruggeri and Strumia \cite{RuggeriStrumia1981}. At the monatomic endpoint the unreduced Gram matrix loses rank; one must first pass to an independent reduced basis, as done below. The physical characteristic velocities $\lambda$ are therefore real and are determined by
\begin{equation}
 \det\left(\mathcal B-\frac{\lambda}{v_T}\mathcal A\right)=0.
 \label{eq:char-det}
\end{equation}
Thus the eigenvalue of the dimensionless matrix pencil is simply $\lambda/v_T$; no separate symbol will be introduced for it.

All the entries of \eqref{eq:time-matrix}--\eqref{eq:flux-matrix} reduce to Gaussian and gamma integrals. In particular, for every integer $r\geq0$,
\begin{equation}
 \left\langle g(\bm\eta)y^r\right\rangle_\alpha
 =\frac{\Gamma(\alpha+r)}{\Gamma(\alpha)}\langle g\rangle_M,
 \qquad
 \langle g\rangle_M=\frac{1}{\pi^{3/2}}
 \int_{\RR^3}g(\bm\eta)e^{-\eta^2}\dd^3\bm\eta.
 \label{eq:gamma-tool}
\end{equation}
This identity is the only ingredient needed for the explicit calculations below.

This construction also displays the principal-subsystem comparison
in the sense of Boillat and Ruggeri \cite{BoillatRuggeri1997ARMA}.
Indeed, if $W$ is any subspace of the independent moment polynomials,
restricting both bilinear forms to $W$ gives a symmetric pencil with
positive-definite time matrix. Its extreme generalized eigenvalues
satisfy the subcharacteristic inequalities
\begin{equation}
 \lambda_{\min,N}^{\mathrm{full}}
 \le \lambda_{\min}^{W}
 \le \lambda_{\max}^{W}
 \le \lambda_{\max,N}^{\mathrm{full}}.
 \label{eq:principal-subsystem-interlacing}
\end{equation}
This is an immediate consequence of the min--max principle for the
symmetric matrix $\mathcal A^{-1/2}\mathcal B\mathcal A^{-1/2}$.
Equivalently, it is the equilibrium, linearized form of the
subcharacteristic property of hyperbolic principal subsystems.
It will be used below with the two-dimensional translational trial
space $\mathcal V_N$. No assertion about a nonlinear principal
subsystem away from equilibrium is needed here.

\emph{Computational verification.} The numerical equilibrium characteristic velocities in Table~\ref{tab:maxspeeds} were computed by assembling the Gram matrices from Gaussian--gamma moments and solving the corresponding symmetric generalized eigenvalue problem. A short Python script (NumPy/SciPy) was used to check these values and to evaluate numerical matrix ranks at the monatomic endpoint. The independence of the reduced moments and the $N=2$ characteristic factorization are established separately by the analytical arguments in the text.

For comparison with the earlier $F,G$ hierarchy, we normalize characteristic velocities by the sound velocity of a monatomic Euler gas,
\begin{equation}
 c_0=\sqrt{\frac{5k_BT}{3m}},\qquad
 \widehat\lambda=\frac{\lambda}{c_0}.
 \label{eq:c0}
\end{equation}
Since $v_T/c_0=\sqrt{6/5}$, equation \eqref{eq:char-det} may of course be expressed equivalently in terms of $\widehat\lambda$.

\subsection{\texorpdfstring{$N=1$: Euler system}{N=1: Euler system}}
For $N=1$ the hierarchy is the five-field Euler system. For a wave in the $x_1$-direction there are three zero characteristic velocities and the two acoustic velocities. In the normalization \eqref{eq:c0},
\begin{equation}
 \widehat\lambda=0\quad\hbox{(multiplicity 3)},\qquad
 \widehat\lambda=\pm\sqrt{\frac{3(D+2)}{5D}}.
 \label{eq:N1-speeds}
\end{equation}
Consequently
\begin{equation}
 \widehat\lambda_{\max}
 =\sqrt{\frac{3(D+2)}{5D}}
 \longrightarrow1
 \qquad(D\to3^+),
 \label{eq:N1-max}
\end{equation}
which is precisely the monatomic Euler value.

\subsection{\texorpdfstring{$N=2$: the polyatomic 15-field system and its Kremer limit}{N=2: the polyatomic 15-field system and its Kremer limit}}
At $N=2$ the characteristic problem has dimension 15. For a wave along $x_1$, isotropy separates the system into zero, transverse and longitudinal sectors. There are seven zero roots; the transverse sector consists of two identical two-dimensional blocks, and the remaining four roots form the longitudinal sector. This decomposition explains the factorization below and avoids treating the resulting roots as unexplained formulas.

Set
\begin{equation}
 L(D)=\frac{2D^2+9D+4}{D(D+2)},\qquad
 K(D)=\frac{3(D+4)}{4D}.
 \label{eq:LK}
\end{equation}
Using \eqref{eq:gamma-tool} in the matrices \eqref{eq:time-matrix}--\eqref{eq:flux-matrix}, and evaluating the determinants of the longitudinal and transverse blocks, gives the following factorization. The factorization follows from exact Gaussian--gamma moments; its characteristic roots were also checked by numerical evaluation of the corresponding Gram matrices. Written in terms of $\widehat\lambda=\lambda/c_0$, the characteristic polynomial is
\begin{equation}
 \widehat\lambda^{7}
 \left[\widehat\lambda^2-\frac{3(D+4)}{5(D+2)}\right]^2
 \left[\widehat\lambda^4-\frac65L(D)\widehat\lambda^2
       +\frac{36}{25}K(D)\right]=0.
 \label{eq:N2-factorization}
\end{equation}
The squared quadratic factor gives a transverse pair, each root having multiplicity two,
\begin{equation}
 \widehat\lambda=\pm\widehat\lambda_T,
 \qquad
 \widehat\lambda_T^2=\frac{3(D+4)}{5(D+2)}.
 \label{eq:transverse}
\end{equation}
The quartic factor gives two longitudinal pairs
\begin{equation}
 \widehat\lambda=\pm\widehat\lambda_-,\qquad
 \widehat\lambda=\pm\widehat\lambda_+,
 \label{eq:longitudinal-pairs}
\end{equation}
where
\begin{equation}
 \widehat\lambda_\pm^2
 =\frac35\left[L(D)\pm\sqrt{L(D)^2-4K(D)}\right].
 \label{eq:longitudinal}
\end{equation}
The largest characteristic velocity is therefore
\begin{equation}
 \lambda_{\max}=c_0\widehat\lambda_+.
 \label{eq:N2-lambda-max}
\end{equation}
At the monatomic endpoint,
\begin{equation}
 L(3)=\frac{49}{15},\qquad K(3)=\frac74,
\end{equation}
and hence
\begin{equation}
 \frac{\lambda_{\max}}{c_0}
 \longrightarrow
 \frac{1}{5}\sqrt{49+\sqrt{826}}
 =1.763408\ldots
 \qquad(D\to3^+).
 \label{eq:N2-lambda-max-limit}
\end{equation}
This limiting value belongs to the independent monatomic ET14 system of Kremer. The full 15-field polynomial \eqref{eq:N2-factorization} has one zero root more than that independent system: the extra zero root is associated with the equation which becomes the identity $H^{(1)}-H_{ll}^{(0)}=0$. Eliminating this redundant equation removes that zero root and leaves all nonzero characteristic velocities unchanged.

\subsection{\texorpdfstring{$N=3$: the $35$-field system and the $30$-field monatomic limit}{N=3: the 35-field system and the 30-field monatomic limit}}
For $N=3$ the polyatomic characteristic problem has dimension 35. In the monatomic limit the five relations \eqref{eq:N3-relations} leave 30 independent fields. The entries of the corresponding equilibrium matrices are again given exactly by \eqref{eq:gamma-tool}; the final diagonalization of the $30$-field monatomic problem is performed numerically. It gives
\begin{equation}
 \frac{\lambda_{\max}}{c_0}
 =2.402446\qquad(D=3).
 \label{eq:N3-lambda-max}
\end{equation}
For orientation, Table~\ref{tab:maxspeeds} reports the maximum equilibrium characteristic velocity for the first three truncation orders. Every entry is the dimensionless quantity $\lambda_{\max}/c_0$.

\begin{table}[ht]
\centering
\caption{Maximum equilibrium characteristic velocity $\lambda_{\max}/c_0$ of the new hierarchy. The first numerical column is the monatomic limit.}
\label{tab:maxspeeds}
\begin{tabular}{ccccc}
\toprule
$N$&$D\to3^+$&$D=5$&$D=7$&$D=10$\\
\midrule
1&1.000000&0.916515&0.878310&0.848528\\
2&1.763408&1.632411&1.565313&1.508490\\
3&2.402446&2.240948&2.152438&2.073355\\
\bottomrule
\end{tabular}
\end{table}

Figure~\ref{fig:speeds} shows the same dependence on $D$. At fixed $N$ the computed maximum velocity decreases as the number of molecular degrees of freedom increases. For $N\ge2$ the values differ from those of the earlier $F,G$ hierarchy because the present full-energy hierarchy contains the additional levels $H_A^{(s)}$ with $s\ge2$.

\begin{figure}[ht]
\centering
\includegraphics[width=.76\textwidth]{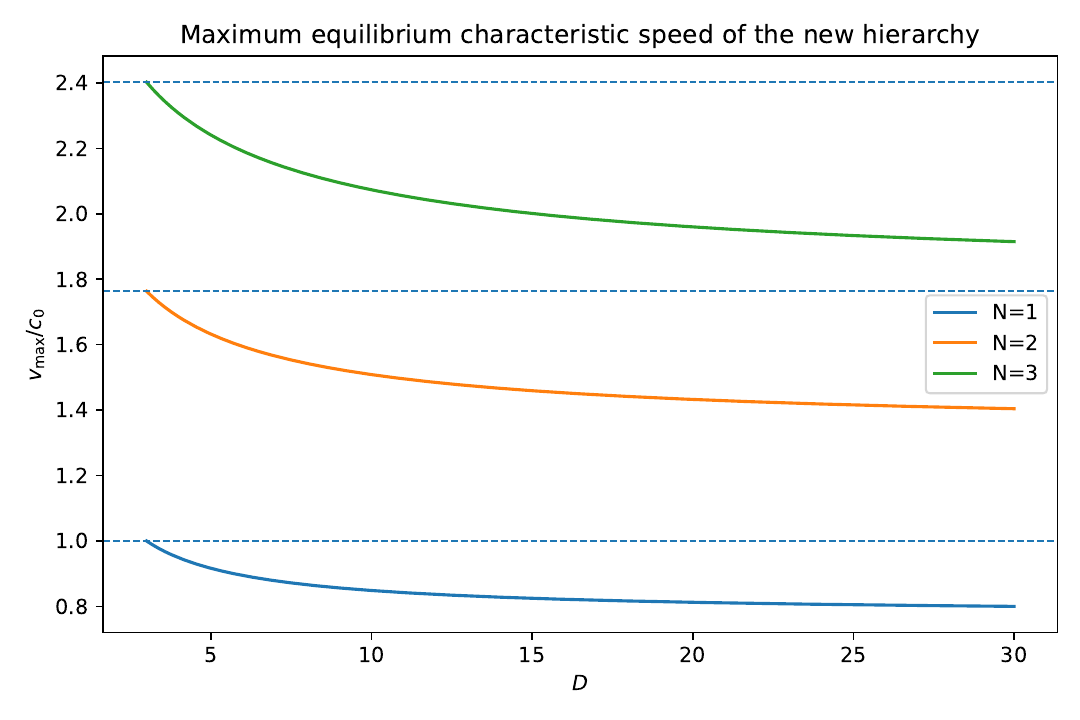}
\caption{Maximum equilibrium characteristic velocity $\lambda_{\max}/c_0$ of the new hierarchy as a function of the molecular degrees of freedom $D$, for $N=1,2,3$. The horizontal dashed lines indicate the corresponding monatomic limits.}
\label{fig:speeds}
\end{figure}

The limit of the full $35$-field spectrum is not identical, as a
multiset, to the spectrum of the reduced $30$-field monatomic system.
It contains three additional zero roots and the additional pair
$\lambda/c_0=\pm\sqrt{27/35}$. Each of the latter values occurs with
multiplicity two in the reduced spectrum and three in the limiting
unreduced spectrum. The structural explanation, valid at every finite
order, is given in Proposition~\ref{prop:spectral-splitting}.

\subsection{Monatomic reduction and the limit of the full spectrum}
There are two distinct spectral operations in the singular limit: taking the characteristic roots after eliminating the dependent monatomic moments, and taking the limit of the entire polyatomic spectrum. The following proposition identifies their precise relation for the polytropic full-energy hierarchy.

\begin{proposition}[Splitting of the limiting characteristic spectrum]
\label{prop:spectral-splitting}
Fix $N\ge2$ and a reference equilibrium state with $\rho>0$ and $T>0$. Let $\operatorname{spec}_{\mathrm{poly},N}(\alpha)$ denote the multiset of characteristic velocities of the order-$N$ polyatomic equilibrium pencil with $\alpha=(D-3)/2>0$, and let $\operatorname{spec}_{\mathrm{mono},N}$ denote the spectrum of the independent monatomic system. Then
\begin{equation}
 \boxed{\displaystyle
 \lim_{\alpha\to0^+}\operatorname{spec}_{\mathrm{poly},N}(\alpha)
 =\operatorname{spec}_{\mathrm{mono},N}
 \ \uplus\ 
 \operatorname{spec}_{\mathrm{poly},N-2}(2).}
 \label{eq:residual-spectrum}
\end{equation}
All velocities are measured at the same fixed reference temperature, and $\uplus$ denotes the union of multisets. The residual block has dimension $k_N$ and gamma-shape parameter $2$, formally corresponding to $D=7$. At order zero its polynomial space consists only of constants, and its spectrum is $\{0\}$.
\end{proposition}

\begin{proof}
\emph{Step 1: a basis adapted to the monatomic kernel.}
Let
\[
 V_N=\operatorname{span}\{\bm\eta^{\bm p}e^s:
                  |\bm p|+s\le N\},\qquad
 R(P)=P(\bm\eta,\eta^2),
\]
with $e$ an independent formal variable. The argument leading to \eqref{eq:reduction-kernel} gives
\begin{equation}
 \ker R=(e-\eta^2)V_{N-2}.
 \label{eq:poly-kernel}
\end{equation}
Choose representatives $\Theta_1,\ldots,\Theta_{r_N}$ of a basis of $V_N/\ker R$ and a basis $S_1,\ldots,S_{k_N}$ of $V_{N-2}$. After the substitution $e=\eta^2+y$, use the same symbols for the resulting polynomials in $(\bm\eta,y)$. The list
\[
 \Theta_1,\ldots,\Theta_{r_N},\quad
 yS_1,\ldots,yS_{k_N}
\]
is a basis of the polyatomic moment space. Let $\widehat{\mathcal A}_\alpha$ and $\widehat{\mathcal B}_\alpha$ be the Gram matrices in this basis. The change from the original basis is invertible and hence preserves the generalized eigenvalues.

\emph{Step 2: rescaling and convergence of the matrix blocks.}
Define the invertible matrix
\[
 T_\alpha=\operatorname{diag}
 \bigl(\mathrm{Id}_{r_N},\alpha^{-1/2}\mathrm{Id}_{k_N}\bigr),
\]
and apply the same congruence to both matrices:
\[
 \widetilde{\mathcal A}_\alpha
 =T_\alpha^{\mathsf T}\widehat{\mathcal A}_\alpha T_\alpha,
 \qquad
 \widetilde{\mathcal B}_\alpha
 =T_\alpha^{\mathsf T}\widehat{\mathcal B}_\alpha T_\alpha.
\]
For every polynomial $F(\bm\eta,y)$, the gamma recurrence gives exactly
\begin{equation}
 \langle yF\rangle_\alpha
 =\alpha\langle F\rangle_{\alpha+1},\qquad
 \frac{\langle y^2F\rangle_\alpha}{\alpha}
 =(\alpha+1)\langle F\rangle_{\alpha+2}.
 \label{eq:gamma-spectral-splitting}
\end{equation}
Thus the upper-left, mixed, and lower-right blocks of $\widetilde{\mathcal A}_\alpha$ have entries, respectively,
\[
 \langle\Theta_i\Theta_j\rangle_\alpha,\qquad
 \sqrt\alpha\langle\Theta_iS_j\rangle_{\alpha+1},\qquad
 (\alpha+1)\langle S_iS_j\rangle_{\alpha+2}.
\]
The corresponding entries of $\widetilde{\mathcal B}_\alpha$ have the additional factor $\eta_1$ in each average. All these polynomial averages are continuous in the indicated gamma-shape parameters; at zero they have the monatomic limit given by \eqref{eq:gamma-tool}. Consequently,
\begin{align}
 \widetilde{\mathcal A}_\alpha&\longrightarrow
 \begin{pmatrix}
  \mathcal A_{\mathrm{mono},N}&0\\
  0&\mathcal A_{N-2}(2)
 \end{pmatrix}
 =:\mathcal A_*,\label{eq:splitting-A}\\
 \widetilde{\mathcal B}_\alpha&\longrightarrow
 \begin{pmatrix}
  \mathcal B_{\mathrm{mono},N}&0\\
  0&\mathcal B_{N-2}(2)
 \end{pmatrix}
 =:\mathcal B_*.\label{eq:splitting-B}
\end{align}
The monatomic blocks are evaluated on the independent restrictions $R(\Theta_i)$; the lower blocks are precisely those of the order-$(N-2)$ polyatomic hierarchy at gamma-shape parameter $2$.

\emph{Step 3: convergence of the generalized eigenvalues.}
The matrix $\mathcal A_{\mathrm{mono},N}$ is positive definite because the restricted polynomials are independent under the Gaussian measure. The matrix $\mathcal A_{N-2}(2)$ is positive definite because the polynomials $S_j$ are independent and the Gaussian--gamma measure has positive density on $\RR^3\times(0,\infty)$. Therefore $\mathcal A_*$ is positive definite. The symmetric matrices
\[
 \widetilde{\mathcal A}_\alpha^{-1/2}
 \widetilde{\mathcal B}_\alpha
 \widetilde{\mathcal A}_\alpha^{-1/2}
 \longrightarrow
 \mathcal A_*^{-1/2}\mathcal B_*\mathcal A_*^{-1/2}
\]
converge in matrix norm. Their ordered eigenvalues consequently converge, including multiplicities. The matrix on the right is block diagonal, with the two spectra stated in \eqref{eq:residual-spectrum}. Since the congruences preserve generalized eigenvalues and $v_T$ is fixed, this proves the assertion for the physical characteristic velocities.
\end{proof}

For $N=2$ the residual block contributes one zero root. For $N=3$ it contributes the five-field Euler spectrum at $D=7$: three zero roots and the pair $\lambda/c_0=\pm\sqrt{27/35}$. These are additional roots of the full limiting pencil, not additional fields of the independent monatomic dynamics.

The analogous distinction was already recognized for the earlier $F,G$ hierarchy by Arima, Mentrelli and Ruggeri \cite[Eq.~(23)]{AMRLincei2014}: the limiting full spectrum of their $(N,N-1)$ polyatomic hierarchy consists of the monatomic $N^+$ spectrum together with the monatomic order-$(N-2)$ spectrum. The mechanism of disappearing equations is common to both theories. The present residual block instead has the shifted gamma-shape parameter $2$ and is polyatomic. Proposition~\ref{prop:spectral-splitting} concerns equilibrium characteristic pencils; it does not by itself establish convergence of solutions of the balance laws, which would also require control of the productions and the initial data.

\subsection{Monotonicity with respect to the molecular degrees of freedom}
\label{subsec:D-monotonicity}
The curves in Fig.~\ref{fig:speeds} suggest that, at fixed truncation
order, the largest characteristic velocity decreases as the number of
molecular degrees of freedom increases. This is not restricted to the
three orders displayed in the figure.

\begin{proposition}[Monotonicity of the maximum equilibrium characteristic velocity]
\label{prop:D-monotonicity}
Let $N\ge 1$ be finite, and consider the polytropic full-energy
hierarchy at a fixed equilibrium temperature and molecular mass.
Write $\alpha=(D-3)/2$ and let $\lambda_{\max,N}(\alpha)$ denote
the largest equilibrium characteristic velocity. At $\alpha=0$,
this symbol refers to the \emph{independent reduced monatomic}
moment system. Then, for $0\le\alpha_1<\alpha_2$,
\begin{equation}
 \lambda_{\max,N}(\alpha_2)
 \le \lambda_{\max,N}(\alpha_1).
 \label{eq:D-monotonicity}
\end{equation}
Furthermore,
\begin{equation}
 \lim_{\alpha\downarrow0}\lambda_{\max,N}(\alpha)
 =\lambda_{\max,N}(0),\qquad
 \sup_{\alpha\ge0}\lambda_{\max,N}(\alpha)
 =\lambda_{\max,N}(0).
 \label{eq:D-monotonicity-endpoint}
\end{equation}
Thus the maximum velocity is non-increasing in $D$ for every finite
$N$; strict decrease is not asserted.
\end{proposition}

\begin{proof}
Use the formal polynomial space $V_N$ and restriction map $R$
defined in Section~\ref{sec:reduction}. For $\alpha>0$, let
$Y_\alpha$ have the Gamma distribution with shape $\alpha$ and
unit scale, independent of the Gaussian variable $\bm\eta$.
At $\alpha=0$ put $Y_0=0$ almost surely, which is the monatomic
Dirac prescription. The largest dimensionless characteristic root
has the Rayleigh characterization
\begin{equation}
 \mu_N(\alpha):=\frac{\lambda_{\max,N}(\alpha)}{v_T}
 =\sup_{P\in V_N\, :\,\langle P^2\rangle_\alpha>0}
 \frac{\langle\eta_1P(\bm\eta,\eta^2+Y_\alpha)^2\rangle_\alpha}
 {\langle P(\bm\eta,\eta^2+Y_\alpha)^2\rangle_\alpha}.
 \label{eq:D-Rayleigh}
\end{equation}
Here the angular brackets denote expectation over $\bm\eta$ and
$Y_\alpha$, with the same Gaussian--gamma weights as in
\eqref{eq:eq-average}. The notation $\langle P^2\rangle_\alpha$
in the supremum abbreviates the denominator in
\eqref{eq:D-Rayleigh}. At $\alpha=0$ the quotient is taken over
the nonzero restrictions $R(P)$, so that it is precisely the
characteristic problem of the independent monatomic system.

For $0\le\alpha_1<\alpha_2$, the Gamma convolution identity gives
\begin{equation}
 Y_{\alpha_2}\ \overset{d}{=}\ Y_{\alpha_1}+Z,
 \qquad Z\sim\operatorname{Gamma}(\alpha_2-\alpha_1,1),
 \label{eq:gamma-additivity}
\end{equation}
with independent summands (and $Y_0=0$ as above). For a fixed
$z\ge0$ define
\begin{equation}
 P_z(\bm\eta,e)=P(\bm\eta,e+z).
 \label{eq:energy-shift}
\end{equation}
The substitution $e\mapsto e+z$ preserves $V_N$: the binomial
expansion of $e^s$ contains only powers $e^k$ with $k\le s$,
and hence does not increase the total polynomial degree used to
define $V_N$. By the variational definition of $\mu_N(\alpha_1)$,
for each $z$ and every $P\in V_N$,
\begin{align}
 &\mathbb E_{\bm\eta,Y_{\alpha_1}}
 \big[\eta_1 P_z(\bm\eta,\eta^2+Y_{\alpha_1})^2\big]
 \nonumber\\
 &\quad\le \mu_N(\alpha_1)
 \mathbb E_{\bm\eta,Y_{\alpha_1}}
 \big[P_z(\bm\eta,\eta^2+Y_{\alpha_1})^2\big].
 \label{eq:conditional-Rayleigh}
\end{align}
This inequality also holds when the right-hand expectation is
zero, since then both sides vanish. Average
\eqref{eq:conditional-Rayleigh} over the independent variable $Z$.
By \eqref{eq:gamma-additivity}, the resulting inequality is exactly
\begin{equation}
 \langle\eta_1 P(\bm\eta,\eta^2+Y_{\alpha_2})^2\rangle_{\alpha_2}
 \le\mu_N(\alpha_1)
 \langle P(\bm\eta,\eta^2+Y_{\alpha_2})^2\rangle_{\alpha_2}.
 \label{eq:D-convex-mixture}
\end{equation}
Taking the supremum over nonzero $P$ proves
\eqref{eq:D-monotonicity}.

It remains to prove continuity at the monatomic endpoint.
Choose a polynomial $P\in V_N$ whose restriction $R(P)$ is a
maximizer of the finite-dimensional monatomic Rayleigh quotient;
such a maximizer exists since the reduced time matrix is positive
definite. Polynomial Gaussian--gamma averages depend continuously
on $\alpha\downarrow0$, and the denominator for this $P$ has a
strictly positive limit. Therefore
\begin{equation}
 \liminf_{\alpha\downarrow0}\mu_N(\alpha)
 \ge \lim_{\alpha\downarrow0}
 \frac{\langle\eta_1P^2\rangle_\alpha}
 {\langle P^2\rangle_\alpha}
 =\mu_N(0).
\end{equation}
The opposite inequality follows from \eqref{eq:D-monotonicity}
with $\alpha_1=0$. This proves \eqref{eq:D-monotonicity-endpoint}.
\end{proof}

The proof uses the convolution structure of the Gamma family and
the stability of $V_N$ under energy translations. It is therefore
specific to this polytropic internal-energy family; monotonicity
for an arbitrary internal-state density is not claimed. It also
shows that the residual spectral block of
Proposition~\ref{prop:spectral-splitting} cannot increase the
limiting largest characteristic velocity beyond that of the reduced
monatomic system.

\subsection{Growth of the maximum characteristic velocity with the truncation order}
The explicit calculations for $N=1,2,3$ suggest a more general question. The new hierarchy contains, for every $N$, the purely translational block $s=0$, namely $H_A^{(0)}=F_A$ with tensorial order $A\leq N$. This observation is important for comparing our result with the classical estimate of Boillat and Ruggeri \cite{BoillatRuggeri1997}.

To avoid an ambiguity of notation, let $N_F$ denote for a moment the highest tensorial degree retained in the classical $F$ hierarchy of Boillat and Ruggeri. In the present full-energy hierarchy the $s=0$ block contains exactly the same translational moments up to degree $N$, hence
\[
 N_F=N.
\]
By contrast, $\bar N=2N$ denotes the highest polynomial degree occurring anywhere in the complete classical hierarchy obtained from the relativistic limit. Thus $N_F=N$ and $\bar N=2N$ have different meanings. The lower bound below concerns the translational truncation order $N_F=N$ and can therefore be compared directly, without any re-indexing, with the Boillat--Ruggeri estimate. The estimate is a direct instance of the subcharacteristic comparison \eqref{eq:principal-subsystem-interlacing} for the equilibrium principal subsystems of Boillat and Ruggeri \cite{BoillatRuggeri1997ARMA}. We include the corresponding two-dimensional Rayleigh-quotient calculation to make the constants and the independence from the internal degrees of freedom explicit.

\begin{corollary}[Boillat--Ruggeri bound for the full-energy hierarchy]
\label{prop:lower-bound}
For every $N\ge1$ and every $D>3$, the largest equilibrium characteristic velocity satisfies
\begin{equation}
 \lambda_{\max,N}
 \geq \sqrt{\frac{(2N-1)k_BT}{m}}.
 \label{eq:new-lower-bound-physical}
\end{equation}
Equivalently,
\begin{equation}
 \boxed{
 \frac{\lambda_{\max,N}}{c_0}
 \geq
 \sqrt{\frac65\left(N-\frac12\right)}.}
 \label{eq:new-lower-bound}
\end{equation}
The same estimate holds for the independent monatomic hierarchy obtained as $D\to3^+$.
\end{corollary}

\begin{proof}
Because the $s=0$ block is present at every order, both functions
\begin{equation}
 \eta_1^{N-1},\qquad \eta_1^N
\end{equation}
belong to the order-$N$ moment space. Restrict the Rayleigh quotient of the symmetric pencil \eqref{eq:char-det} to
\begin{equation}
 \mathcal V_N=\operatorname{span}\{\eta_1^{N-1},\eta_1^N\}.
\end{equation}
These two trial functions do not depend on the internal variable $y$. Hence the normalized $y$-integration in \eqref{eq:eq-average} is exactly one for every $\alpha>0$.

Let
\begin{equation}
 M_{2r}=\frac1{\sqrt\pi}\int_{-\infty}^{\infty}
 \eta_1^{2r}e^{-\eta_1^2}\dd\eta_1
 =\frac{(2r-1)!!}{2^r}
 \label{eq:Gaussian-moments}
\end{equation}
be the one-dimensional even Gaussian moments. On $\mathcal V_N$ the matrices \eqref{eq:time-matrix}--\eqref{eq:flux-matrix} reduce to
\begin{equation}
 \mathcal A_N=
 \begin{pmatrix}M_{2N-2}&0\\[2mm]0&M_{2N}\end{pmatrix},
 \qquad
 \mathcal B_N=
 \begin{pmatrix}0&M_{2N}\\[2mm]M_{2N}&0\end{pmatrix}.
 \label{eq:trial-matrices}
\end{equation}
The corresponding restricted characteristic equation is
\begin{equation}
 \det\left(\mathcal B_N-\frac{\lambda}{v_T}\mathcal A_N\right)=0,
\end{equation}
and therefore
\begin{equation}
 \frac{\lambda}{v_T}
 =\pm\sqrt{\frac{M_{2N}}{M_{2N-2}}}
 =\pm\sqrt{N-\frac12}.
 \label{eq:trial-root}
\end{equation}
By the min--max characterization of the eigenvalues of a symmetric generalized problem, the largest eigenvalue of the complete system is not smaller than the largest eigenvalue on the subspace $\mathcal V_N$. Consequently
\begin{equation}
 \lambda_{\max,N}
 \ge v_T\sqrt{N-\frac12}
 =\sqrt{\frac{(2N-1)k_BT}{m}},
\end{equation}
which is \eqref{eq:new-lower-bound-physical}. Division by $c_0$ gives \eqref{eq:new-lower-bound}. At the monatomic endpoint the same two trial functions remain independent members of the surviving $s=0$ block, so the argument is unchanged.
\end{proof}

Combining Corollary~\ref{prop:lower-bound} with
Proposition~\ref{prop:D-monotonicity} gives the two-sided comparison
\begin{equation}
 \sqrt{\frac65\left(N-\frac12\right)}
 \le \frac{\lambda_{\max,N}(D)}{c_0}
 \le \frac{\lambda_{\max,N}(3)}{c_0},
 \qquad N\ge1,\quad D\ge3.
 \label{eq:two-sided-characteristic-bound}
\end{equation}
The right-hand inequality alone does \emph{not} imply the
Boillat--Ruggeri lower bound: the latter follows independently from
the translational principal-subsystem argument above.

It follows immediately from Corollary~\ref{prop:lower-bound} that
\begin{equation}
 \lim_{N\to\infty}\lambda_{\max,N}=+\infty
 \label{eq:lambda-max-infty}
\end{equation}
for both the polyatomic hierarchy and its monatomic limit. Thus each finite classical moment system has finite characteristic velocities, while no uniform upper bound survives when the truncation order tends to infinity.

\begin{remark}[Why the bound is independent of the internal degrees of freedom]
\label{rem:universal-bound}
Neither $D$ nor $a$ appears in \eqref{eq:new-lower-bound}. The reason is visible directly in the proof: the trial space $\mathcal V_N$ lies entirely in the purely translational block $s=0$, and therefore the internal-energy integration contributes only its normalization, equal to one. The same lower bound consequently holds for the polyatomic hierarchy and for its monatomic limit. More generally, the same argument applies to classical kinetic or MEP moment closures which contain these translational moments and whose equilibrium tangent symbol is represented by the same Gram forms $\mathcal A(\Psi,\Phi)=\langle\Psi\Phi\rangle$ and $\mathcal B(\Psi,\Phi)=\langle\eta_1\Psi\Phi\rangle$, with the usual Maxwellian translational equilibrium. The argument is then insensitive to the particular equilibrium law assigned to the internal energy. In this precise, restricted sense the estimate is universal. The statement is specific to the classical theory; in the relativistic theory characteristic velocities remain bounded by the speed of light.
\end{remark}

For the equilibrium linearized pencil, the lower bound \eqref{eq:new-lower-bound} is therefore not merely analogous to the classical Boillat--Ruggeri inequality: with the identification $N_F=N$ just explained, it is exactly the same inequality proved in \cite{BoillatRuggeri1997}. Arima, Mentrelli and Ruggeri later proved that the same bound holds for the earlier polyatomic $F,G$ hierarchy independently of the molecular internal degrees of freedom \cite{AMR2014}. The present proof shows that it survives unchanged for the new full-energy hierarchy.

The relativistic situation is qualitatively different. For the monatomic relativistic moment hierarchy, Boillat and Ruggeri proved that the maximum characteristic velocity remains below the light velocity and approaches $c$ as the number of moments tends to infinity \cite{BoillatRuggeri1999CMT,BoillatRuggeri1999JMP}. The relativistic polyatomic full-energy theories analysed at finite truncation order are causal \cite{ACPREntropy2022,ACPR2022}. It is therefore natural to expect an analogous saturation at $c$ for the general relativistic polyatomic full-energy hierarchy. To our knowledge this limiting result has not yet been established for arbitrary truncation order and will be the subject of a future investigation.

\section{Conclusions}

For the polytropic density of internal states $\phi(I)=I^a$, with $a=(D-5)/2$, the monatomic limit $D\to3^+$ is the singular limit $a\to-1^+$. Theorem~\ref{thm:dirac} shows that, in normalized kinetic integrals, this limit is exactly equivalent to replacing the internal-state integration by the unit-Dirac prescription concentrated at $I=0$. The first correction is given explicitly by \eqref{eq:firstcorrection}.

Applied to the classical full-energy hierarchy \eqref{eq:H-hierarchy}, this rule gives the trace relations \eqref{eq:trace-relations} simply by setting $I=0$. These relations reduce the $d_N$ polyatomic fields to the $r_N$ independent monatomic fields given in \eqref{eq:mono-count}. The first nontrivial example is $N=2$: the polyatomic 15-field system loses the dynamic pressure and reduces to Kremer's monatomic ET14, with the additional scalar identified by
\[
 F_{llkk}=15\frac{p^2}{\rho}+\Delta.
\]
For $N=3$, five independent relations reduce the 35 polyatomic fields to 30 monatomic fields. The distinction between the reduced characteristic spectrum and the limit of the full polyatomic pencil is made explicit in Proposition~\ref{prop:spectral-splitting}, which also identifies its residual spectral block and connects the mechanism with the earlier $F,G$ theory.

For the equilibrium characteristic problem we have used the linearized entropy-projection (MEP tangent) closure, for which the matrices are Gaussian and gamma averages as derived in Section~\ref{sec:speeds}. This calculation is distinct from a claim of nonlinear MEP integrability for arbitrary truncation order. We have written explicitly the cases $N=1$ and $N=2$, and evaluated the maximum velocity for $N=3$. More generally, Proposition~\ref{prop:D-monotonicity} proves that, for every finite $N$ in the polytropic family, the maximal equilibrium characteristic speed is non-increasing in $D$ and takes its largest value at the reduced monatomic endpoint. This result relies on the Gamma convolution property and is distinct from the universal lower bound. By the symmetric-hyperbolic main-field construction and the principal-subsystem interlacing of Ruggeri--Strumia and Boillat--Ruggeri, respectively, Corollary~\ref{prop:lower-bound} gives
\[
 \frac{\lambda_{\max,N}}{c_0}
 \geq
 \sqrt{\frac65\left(N-\frac12\right)}.
\]
Hence every finite truncation has finite characteristic velocities, but
\[
 \lambda_{\max,N}\longrightarrow+\infty
 \qquad(N\to\infty)
\]
both for the polyatomic hierarchy and for its monatomic limit. The proof uses only the purely translational $s=0$ block, which explains why the bound is independent of the number of internal degrees of freedom and, within the class specified in Remark~\ref{rem:universal-bound}, of the detailed internal equilibrium law.

Finally, the algebraic reduction of the moments must be distinguished both from propagation of the limiting constraints and from convergence of polyatomic solutions. For a family with $D>3$ the defect combinations are not set identically to zero; rather, prepared initial data must approach the monatomic constraint manifold as $D\to3^+$. For the limiting $N=2$ monatomic data this includes $\Pi(0,\bm x)=0$, consistently with the earlier singular-limit analysis. Propagation requires the closed limiting fluxes and productions, while a complete nonlinear convergence theorem requires additional uniform estimates; both questions lie outside the scope of the present paper.

\section*{Acknowledgements}
This work was carried out within the activities of GNFM/INdAM (Gruppo Nazionale per la Fisica Matematica).

\section*{Declaration of competing interest}
The author declares that he has no known competing financial interests or personal relationships that could have appeared to influence the work reported in this paper.

\section*{Data availability}
No external datasets were used in this study. The Python script used to verify the numerical results is available from the author upon reasonable request.

\section*{Declaration of generative AI and AI-assisted technologies in the manuscript preparation process}
During the preparation of this manuscript, the author used ChatGPT (OpenAI) for assistance with language and LaTeX presentation, and for checking mathematical and numerical calculations. All AI-assisted suggestions were critically reviewed, and the author takes full responsibility for the scientific content and conclusions of the work.

\end{document}